\documentclass[sigconf]{acmart}

\copyrightyear{2019}
\acmYear{2019} 

\usepackage{booktabs} % For formal tables
\usepackage{graphicx}
\usepackage{amsmath}
\usepackage{algpseudocode}
\usepackage[ruled]{algorithm2e} % For algorithms

\SetCommentSty{mycommfont}
\usepackage{amssymb}
\usepackage{subcaption}
\usepackage{color}
\usepackage{array}
\usepackage{multirow}
\usepackage[space]{grffile}
\usepackage{breqn}
\usepackage{float}
\usepackage{amsthm}

\newtheorem*{theorem*}{Theorem}
\def\BibTeX{{\rm B\kern-.05em{\sc i\kern-.025em b}\kern-.08emT\kern-.1667em\lower.7ex\hbox{E}\kern-.125emX}}

\begin{document}

%
% The "title" command has an optional parameter, allowing the author to define a "short title" to be used in page headers.
\title{Leveraging Motifs to Model the Temporal Dynamics of Diffusion Networks}

%
% The "author" command and its associated commands are used to define the authors and their affiliations.
% Of note is the shared affiliation of the first two authors, and the "authornote" and "authornotemark" commands
% used to denote shared contribution to the research.
\author{Soumajyoti Sarkar}
\email{ssarka18@asu.edu}
\affiliation{%
	\institution{Arizona State University}
	\city{Tempe}
	\state{AZ}
	\postcode{85821}
}

\author{Hamidreza Alvari}
\email{halvari@asu.edu}
\affiliation{%
	\institution{Arizona State University}
	\city{Tempe}
	\state{AZ}
	\postcode{85821}
}

\author{Paulo Shakarian}
\email{shak@asu.edu}
\affiliation{%
	\institution{Arizona State University}
	\city{Tempe}
	\state{AZ}
	\postcode{85821}
}

%
% By default, the full list of authors will be used in the page headers. Often, this list is too long, and will overlap
% other information printed in the page headers. This command allows the author to define a more concise list
% of authors' names for this purpose.
\begin{abstract}
	Information diffusion mechanisms based on social influence models are mainly studied using likelihood of adoption when active neighbors expose a user to a message. The problem arises primarily from the fact that for the most part, this explicit information of who-exposed-whom among a group of active neighbors in a social network, before a susceptible node is infected is not available. In this paper, we attempt to understand the diffusion process through information cascades by studying the temporal network structure of the cascades. In doing so, we accommodate the effect of exposures from active neighbors of a node through a network pruning technique that leverages network motifs to identify potential infectors responsible for exposures from among those active neighbors. We attempt to evaluate the effectiveness of the components used in modeling cascade dynamics and especially whether the additional effect of the exposure information is useful. Following this model, we develop an inference algorithm namely~\textsc{InferCut}, that uses parameters learned from the model and the exposure information to predict the actual parent node of each potentially susceptible user in a given cascade. Empirical evaluation on a real world dataset from Weibo social network demonstrate the significance of incorporating exposure information in recovering the exact parents of the exposed users at the early stages of the diffusion process.
\end{abstract}
% The code below should be generated by the tool at
% http://dl.acm.org/ccs.cfm
% Please copy and paste the code instead of the example below.
%



\maketitle

\section{Introduction}
Information, genuine or spurious, spreads over time. Modeling information diffusion to understand how information propagates among a population of individuals has been a very popular research problem in the social network and machine learning community. With availability of real world data, these models can be manipulated to overcome several challenges \cite{aral2009distinguishing}. Most models in information diffusion work on the premise of influence from active neighbors: the effect of multiple active neighbors sequentially or collectively in an arbitrary order selectively enables the spread of information. These models implicitly assume the effect of multiple exposures from these active neighbors. However it does not account for how well a neighbor is structurally or temporally placed in the social network with respect to the current contagion that could eventually decide the role of a particular neighbor in influencing a user \cite{Katona11}.

In this paper, we try to address the notion of ``impactful" exposures in the overall framework of influence and the diffusion mechanism where we would elaborate what we mean by ``impactful" in the technical section - briefly, these are exposures from among a \textit{subset} of active neighbors of a user that are structurally and temporally positioned to influence that user in an ongoing contagion.  To this end, we study the spread of information by understanding the evolution of information cascades\cite{Kempe2003} which are simply messages reshared by multiple individuals in a social network. To test whether the information obtained from these ``impactful" exposures as opposed to all active neighbors is useful, we use the exposure information in a system that models these information cascades. However, for evaluating how this model with these exposures compares to the existing methods, we devise a prediction problem that infers the parents (the user which ultimately infects a user) of a susceptible user and test these models on this prediction problem. This prediction problem is similar to the link prediction problem \cite{Liben2007} although we consider the temporal setting within a more constrained setup. Summarizing, the main two less studied shortcomings of the existing information diffusion models that we focus on are as follows:

\begin{enumerate}
	\item The majority of the models that consider either social influence or viral propagation processes through active neighbors cannot identify neighbors who were responsible for exposure or consider all active neighbors of a susceptible user as nodes which exposed that user. In fact, one of the popular ways of assessing the impact of exposures is to consider the collective influence of \textit{all active neighbors} of a susceptible user and use it in a threshold model \cite{kempe2003maximizing,goyal2010learning}.
	
	\item Secondly, network inference methods \cite{Gomez20103,Gomez2011} that attempt to recover edges in a static network or in a temporal setting to predict links ahead of time assume complete observations from the start of diffusion (of the cascade) till a certain time. We attempt at inferring the activation sources in a constant evolving manner by using temporal networks and the current available information instead of all information from the beginning. 
\end{enumerate}

Studying the effect of filtered exposures facilitates understanding the diffusion process and thereafter inferring the parent nodes in the following way: (1) first this allows us to test the hypothesis that all exposures from active neighbors are not responsible in the collective influence of an infected node and since exposures precede infection, this paves the way toward inferring the parent node of an infected user from among the active neighbors by the process of elimination, (2) secondly  we leverage network motifs \cite{Milo824} which are recurring subgraphs in a network occurring significantly, for identifying exposure nodes for a susceptible user with respect to information cascades and the social network information (the past historical data). The advantage of using motifs as would be elaborated further is that it takes advantage of the position of the neighbors within the cascade as well as the past historical social network to assess whether an active neighbor of a susceptible user can even expose a user to the message and therefore be a potential candidate for its parent. Following this, our main contributions are as follows:

%Keeping these in mind, the broad problem we solve is the following: we take as an input a partially observed information cascade with the following information of each reshare in the cascade: (1) user who activated the resharer which we call \textit{parent} nodes , (2) users who adopted the information, and (3) the time that the information was adopted. Given this information from the recent past and a sample of nodes that we know would get infected shortly, \textit{we try to infer the exact parent node for each of the sample nodes ahead of time}. Intuitively, at each stage of the cascade, we try to infer the cascade diffusion links. 

\begin{itemize}
	\item We devise a method that uses network motifs to prune out the potential exposure nodes for a susceptible user at different stages of an information cascade.
	\item We use this additional information from the exposure nodes to model the evolution of the information cascade based on techniques from survival analysis introduced in~\cite{Gomez2011}.
	\item To test the effectiveness of this exposure information and therefore the model of diffusion, we conduct a suit of experiments on a real world dataset from Weibo to predict the parent (from whom users reshare a piece of message) of a user. We make the following observations: (1) taking exposure into account helps in the parent inference process in the early stages of the cascade compared to the models with just temporal information (2) the parent inference using the cascade topology becomes difficult as the cascade progresses, suggesting that the structure loses any specific organization thus failing to use any additional information on motifs.
\end{itemize}

The rest of the paper is as follows: we first give a brief description of the technical preliminaries on information cascades and diffusion in Section~\ref{sec:tech_prelim}, followed by the model of information diffusion demonstrating the exposure effect in Section~\ref{sec:model}, the dataset description in Section\ref{sec:dataset} and finally the results discussion in Section~\ref{sec:results}.
 
\section{Technical Preliminaries} \label{sec:tech_prelim}
We shall represent historical social network information by a directed graph $G_{\mathcal{H}} = $ $(V_{\mathcal{H}} ,  E_{\mathcal{H}} )$. In this graph, $V_{\mathcal{H}}$ denotes the individuals who were either involved in historical diffusion processes (thus building historical  diffusion networks) or are part of the follower-followee network. We denote by $e=(i, j)$ $ \in$ $E_{\mathcal{H}}$ to imply that an information can be propagated from user $i$ to user $j$ given that  a similar diffusion took place in history or by virtue of the follower-followee network. These relationships are known apriori using observed or available data.  Further, an event or a message can originate from node $u$ $\in V_{\mathcal{H}}$ and can be reshared by a neighbor of $u$ in the network $G_{\mathcal{H}}$. This produces a cascade $C$ of reshares which comprises individuals $V_C$ = $\{v_1, v_2, \ldots, v_{\mathcal{R}}\}$ ordered by their corresponding chronologically sorted infection times $\tau_C$ = $\langle t_1, t_2 \ldots t_{\mathcal{R}} \rangle$, where $\mathcal{R}$ denotes the number of reshares/infections of $C$ or size of the cascade. We refer to these individuals as adopters of $C$. In addition, we assume that ultimately each individual $v \in V_C$ is infected by exactly one existing adopter of $C$ which we refer to as its \textit{parent} in the cascade lifecycle. This assumption relies on the observation that a user can reply to or retweet one piece of information multiple times, however it posts it the first time being  activated at a single point in time. We denote the underlying network produced by the participants of a cascade $C$ by $G_C$ = $(V_C, E_C)$ where an edge $e=(u, v) \in$ $E_C$ denotes that either $v$ reshared $C$ from an adopter $u$ which we term $parent(v)$ or there is an edge $e$ in $E_{\mathcal{H}}$. In simple terms, we add the  social/historical diffusion network links from $G_{\mathcal{H}}$ to create an augmented cascade network. We will use the set $Parents(V')$ to denote the set of parents for each node in $V' \in V_C$. We denote the time at which node $v$ reshared a cascade $C$ as $t_{v, C}$. Also, we will drop the subscript $C$ from the notations when an operation described on a cascade is applicable for all cascades. We note that the presence of social/historical network edges introduces cycles in the structure of $G_C$ which otherwise would exhibit a tree structure inherent to the property of cascades.

We use the notation $\tau'$ to denote a subsequence of $\tau$. For the purposes of notation, we will consider the subsequence $\tau''$ as the succeeding subsequence of $\tau'$ for any two subsequences $\tau'$ and $\tau''$ of $ \tau$. In our work, we create the sequence of subsequences $\tau$ = $\langle \tau'_1, \ldots, \tau'_{\mathcal{Q}} \rangle$, ordered by the starting time of each subsequence, where we denote $ \mathcal{Q}$ to be the number of subsequences for $C$, which would vary for different cascades. The method of splitting the cascade $C$ into a sequence of such subsequences $\tau'$ is described in the following section. 

\begin{figure}[!t]
	\centering
	\includegraphics[width=8.5cm, height=3cm]{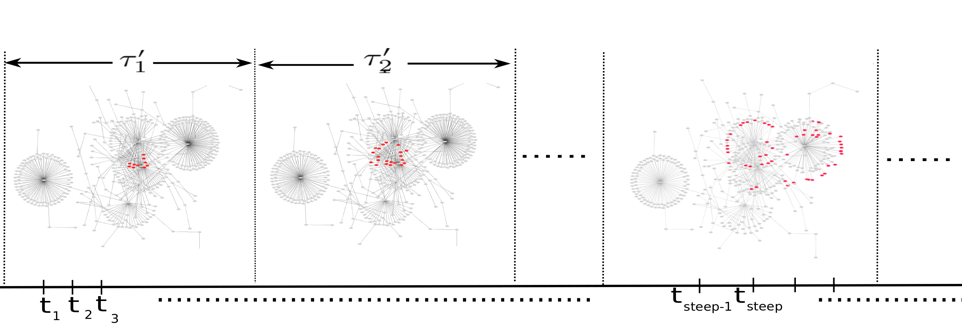}
	\caption{The cascade partitioning procedure. The likelihood for the cascades is computed over this incrementally evolving network. The number of nodes activated in each subsequence $\tau'$ is kept fixed. Red nodes denote new activated nodes in each subsequence.}
	\label{fig:net_model}
\end{figure}

\subsection*{Temporal network evolution}
\label{sec:net_analysis}
In context of $G^{\tau}$,  we denote $G^{\tau'}$ = $(V^{\tau'}, E^{\tau'})$ as the subgraph of $G^{\tau}$, where $V^{\tau'}$ denotes the individuals who reshared a cascade $C$ in the time subsequence $\tau'$ and  $E^{\tau'}$ denotes the set of edges $e=(u, v)$ where the resharing from $u$ to $v$ happened in the time subsequence $\tau'$ or there was an interaction in the historical diffusion period indicated by the presence of $e$ in $E_D$.  For the subsequences, the following conditions hold: (1) $|V^{\tau'_i}| = |V^{\tau'_j}|$ and (2) $E^{\tau'_i} \cap E^{\tau'_{j}}$ = $\emptyset$, $\forall i \neq j, \in [0, \mathcal{Q}]$. We note that the condition $|\tau'_i|$ $\neq |\tau'_j|$ may or may not hold for any $i \neq j$, that is to say the time range spanned by the subsequences in itself may differ depending on the time taken by $G^{\tau'}$ to form the network. In our work, we select and keep $|V^{\tau'}|$ fixed for every cascade in our corpus. Since we analyze the subsequences in a sequence that depicts the evolution of the cascade over time, the advantage of selecting this subsequence node set size a-priori is that we can avoid retrospective observation of the entire cascade lifecycle and incrementally progress with the network analysis using motifs over intervals until we reach $N_{inhib}$. This gives us the freedom to be agnostic about the final cascade size $|V^{\tau}|$ or $T_C$ for selecting the subsequence size $|V^{\tau'}|$ or span $|\tau'|$. 

A \textit{temporal representation} of a cascade is denoted by a sequence of overlapping subsequences $\mathcal{N}$ = $\langle N_1, \ldots, N_{\mathcal{S}} \rangle$ such that the following conditions hold: $V^{N_i}$ = $V^{\tau'_{i-1}}$ $\cup$ $V^{\tau'_i}$, and $E^{N_i}$ = $E^{\tau'_{i-1}}$ $\cup$ $E^{\tau'_i}$ $\forall$  $i \in [1, \mathcal{Q}]$. We perform network analysis on each $N $ where we drop the index subscript when we generalize the analysis for all subsequences for all cascades. Such a temporal representation $\mathcal{N}$ helps us in avoiding disjoint subsequences for network analysis and replicates a sliding window approach. This method to partition the cascade into subsequences would help us in two ways: (i) having the network topology of any subsequence $G^{\tau'}$, we can attempt to find $parent(v)$,  $ v \in V^{\tau''}$ from among nodes in $V^{\tau}$, where $\tau''$ denotes the subsequences immediately exceeding $\tau'$, and (ii) we can prune the nodes in $G^{\tau'}$ to compute the \textit{exposure nodes} temporally relevant to each node in $G^{\tau''}$, for which we use network motifs that would be described later. In this way, instead of considering all information from the start of the cascade, we only consider the recent past to predict the parents of future nodes. An example of the cascade network procedure that governs the network evolution is shown in Figure~\ref{fig:net_model}. 

\begin{figure}[!t]
	\centering
	\includegraphics[width=7.5cm, height=3.5cm]{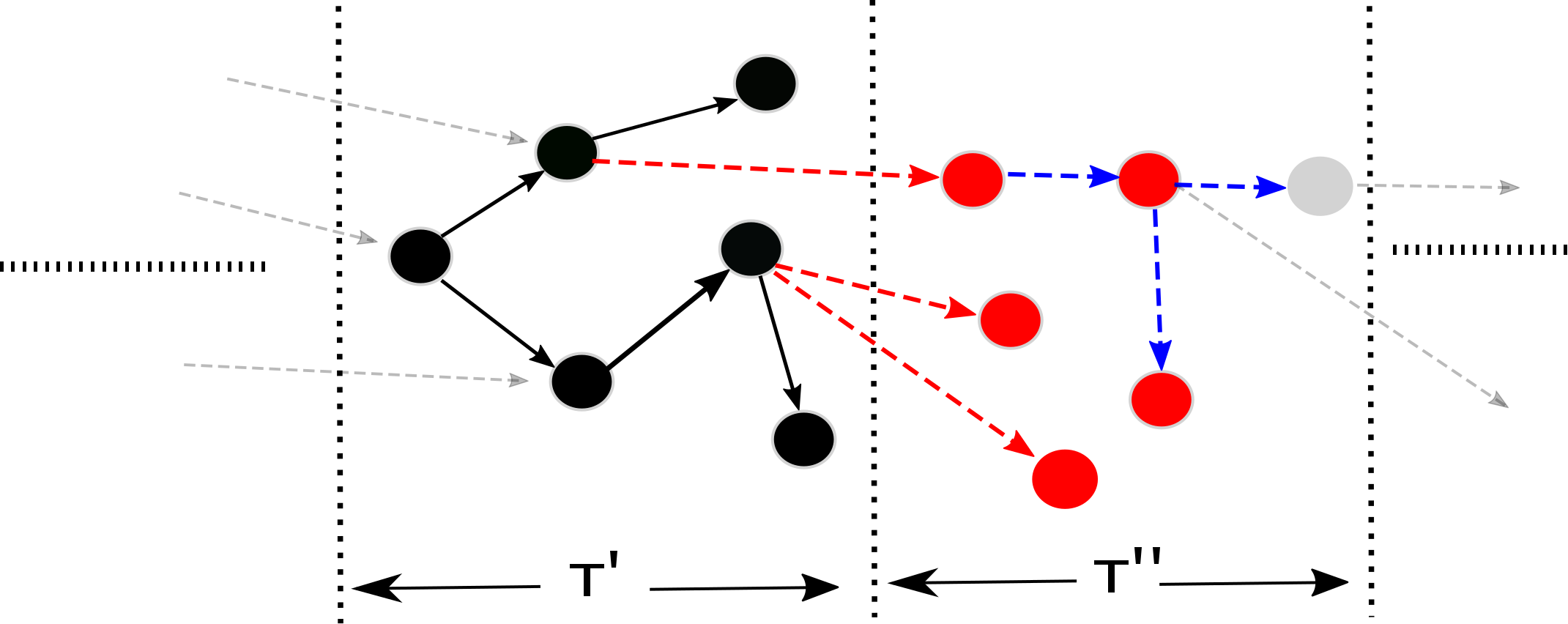}
	\caption{The problem statement setting. Black nodes represent the observed portion of the cascade in a subsequence $\tau' \in \tau_C$ in terms of diffusion. Red nodes denote the sample of nodes with known attributes which participates in $\tau''$ but with unknown $parent$ nodes. Blue edges are those which are ignored in the reconstruction process. Observe that the potential parents form the ends of cut edges between $V^{\tau'}$ and $V^{\tau''}$. We infer these red edges.}
	\label{fig:pr_state}
\end{figure}

\section{Problem Statement}
Given a temporal network in the current time phase of the cascade lifecycle and a list of susceptible and \textit{to-be} infected nodes in the next timestamp (Figure~\ref{fig:pr_state}), our goal is to discover the set of nodes from this current temporal network that are going to activate these susceptible nodes. The problem is formally defined as follows: 
\textit{Given a graph $G^{\tau'}_C$ = $(V^{\tau'},  \ E^{\tau'})$, $V^{\tau''}$, such that $\tau', \tau^{''} \in \tau_C$, $\tau''$ succeeds $\tau'$ in the set of subsequences, our goal is to find $parent(v)$ $\forall v$ $\in V^{\tau''}$ with the condition $parent(v) \in V^{\tau'}$.}
The constraint that the set of parents for nodes in $V^{\tau''}$  must be part of $V^{\tau'}$ would not always hold as shown in Figure~\ref{fig:pr_state}. Note that we do not predict whether the susceptible nodes would be infected but as a retrospective measure, we attempt to find their parents. Our problem statement in that respect introduces two new  directions for research:
\begin{enumerate}
	\item It sets this problem apart from general network inference where network edges are inferred irrespective of whether there is a chance of having a node as a parent of another node through diffusion mechanism governed by network topology.
	\item A single parent inference procedure ultimately paves the way towards diffusion path reconstruction in cascades since as the stages (or subsequences) unfold over time, each node can be traced back to its parent and the diffusion tree can be reconstructed. It is not necessary to have a fully connected tree as considered in \cite{Rozenshtein} to be able to construct the tree in a forward progressing manner.
\end{enumerate}
We do note that this temporal parcellation of the cascade introduces different lags for parent consideration for different nodes (since the subsequences are not compact and bounded manually by \textit{time} but by node size of subsequences). However empirically, we found that the time difference between a node's activation time and it's parent's activation time follows a power law where most users reshare a message from parents who reshared the same message not far before. Therefore using information from only the recent past to predict the parent generalizes well in the real-world scenario.

\section{Model Learning and Inference}
\label{sec:model}
In this section, we describe in details the 3 components that lay the roadmap: finding the potential exposure nodes using motifs to parametrically modeling the diffusion process to finally extending the learned model parameters and the motifs to infer the parents of future nodes in the cascade lifecycle. To model the likelihood of diffusion between the set of nodes in a subsequence $\tau''$,  we would use the information from the preceding subsequence $\tau'$. We proceed through the stages $[1, \mathcal{Q_C}]$ to aggregate the likelihood for $C$. We summarize the 3 components as follows: (1) we extract the parent node (available from data) and prune the exposure nodes for each node $v$ that participated in the cascade $C$ prior to its own infection, (2) we use the parent and exposure nodes for all participants to compute the likelihood of the cascades relying on previous techniques used in survival analysis and formulate a convex objective function and we use an optimization algorithm for learning the parameters used in the likelihood function and (3) finally, given a test subsequence of nodes in $\tau'$ and the sample of nodes from $\tau''$, we devise an inference algorithm based on the learned model parameters to infer $Parents(V^{\tau''}) \in V^{\tau'}$. We note that the parametric model helps in the prediction problem on unseen cascades and the network motifs are used for pruning out exposure nodes. We combine these two to infer the parent of a susceptible node - without the model and just the information about exposure nodes, we would have no mechanism to ultimately select one node among the filtered exposure nodes.

\begin{figure}[!t]
	\centering
	\includegraphics[width=7.5cm, height=1cm]{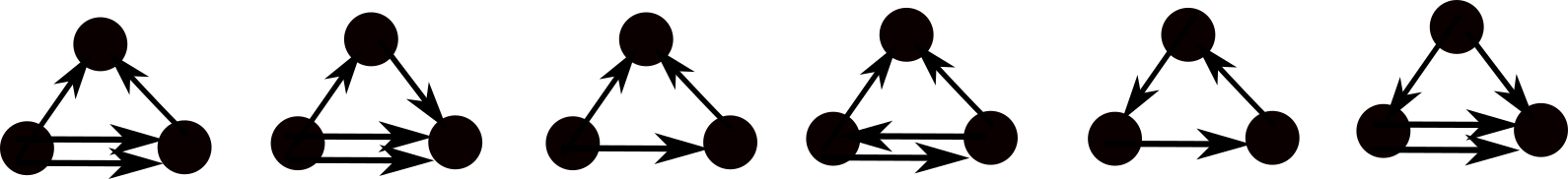}
	\caption{Motif patterns used in computing the exposure nodes. Note that some of the patterns contain parallel nodes - one of them denoting a social / historical diffusion link apart from the cascade diffusion link. These motifs have been chosen based on significance in the frequencies obtained in a few selected cascades - there are many ways of assessing significance, like z-scores, count. We use count (these appear more frequently than 1000 other motif patterns) from a few observed cascades to select these patterns.}
	\label{fig:motifs_pat}
\end{figure}

\subsection{Extracting parent and exposure nodes}
As a first step in building the model as part of the training stage, we iteratively take two consecutive subsequences $\tau', \tau^{''}$ in a cascade lifecycle in order until the end of the cascade lifecycle. Furthermore, assuming that we know \\ $Parents(V^{\tau''})$ (as ground truth  labels from the data) during each stage $\tau''$, our main goal in this step is to infer for each node $v \in V^{\tau''}$, the set of exposure nodes $exp(v) \in V^{\tau'} \setminus parent(v) $. In order to prune the nodes which are relevant to $v$ and which could have potentially exposed $v$ to $C$, we use network motifs to extract subgraphs in the network $G^{\tau''} \cup G^{\tau'}$. A network motif $M=(V_M, E_M)$ is defined as an induced subgraph of $G$, such that $\forall e $ $\in E_M$, it is always that $e$ $\in V_G$. To prune the exposure nodes $exp(v)$ for a node $v$ using network motifs, we consider the following three steps: (1) We first extract all the instances of 3-size motifs containing $v$ from $G^{\tau''} \cup G^{\tau'}$, whose patterns match one of the patterns shown in Figure~\ref{fig:motifs_pat}, (2) following this, for each $v \in V^{\tau''}$, we prune those motif instances among all such that each instance contains the $parent(v)$ among the other two nodes apart from $v$ with the remaining third vertex $v'$ $\in V^{\tau'}$. We collect the list $\mathcal{L}$ containing the set $\{v'\}$ from all pruned instances to form an initial set of potential exposure nodes for $v$, (3) finally following the work of \cite{Milo824}, which uses motif nodes as input elements in a gating function  to control the output of another node, we use these potential exposure nodes $\mathcal{L}$ together with $parent(v)$ as input to a message passing switch controlling signals to the output vertex $v$. For that we first multiply the product of time differences $\frac{1}{t_v - t_w}$ $\forall w \in \mathcal{L}$ (Note that $t_v - t_w$ always $>$ 0). We start removing nodes from the list of the sorted time differences such that the final product of the time differences crosses a threshold. We denote this remaining list of nodes as $exp(v)$. Essentially, this last step helps prune out nodes which do not satisfy the AND condition for a positive output in a gate where instead of boolean values we consider real values in a continuous space. We present a detailed algorithm for this step in the Appendix section. Apart from the gating mechanism that motifs are typically used for, the fact that the network motifs appear at frequencies much higher than expected at random \cite{Alon2002}, suggests that the frequent appearance of a node $v$ along with its potential exposure nodes in certain patterns is a result of  its frequent interactions with those nodes in the past and a possible exposure effect from among other peers. At the end of this step, we have a set of nodes $C$ = $\{v_1, v_2, \ldots, v_{\mathcal{R}}\}$ and the exposure nodes $\zeta_C$ = $\{\exp(v_1), \ldots, exp(v_{\mathcal{R}})\}$.

\subsection{Model of diffusion}
In order to model the cascade likelihood, we borrow concepts from survival analysis which has been widely used for computing the diffusion likelihood of the cascades \cite{Gomez2011}. We first briefly describe a few concepts before formulating the likelihood objective function.

\noindent \textbf{Survival Analysis.} Let $T$ be a non-negative random variable denoting the time when an event happens or an individual is infected. Let $f(t)$ be the probability density function of $T$ for which $F(t$) = $Pr(T \leq t)$ = $ \int_0^T f(x) \ dx$ is the its cumulative distribution function. The survival function $S(t)$ gives the probability that an event does not happen up to time $t$, $S(t) = Pr(T \leq t) = 1 - F(t) = \int_t^\infty f(x) \  dx $. Thus, $S(t)$ is a continuous and monotonically decreasing function with $S(0) = 1$ and $S(\infty)$ = $\lim_{t\to\infty} S(t)$ = 0.
Given $f(t)$ and $S(t)$, the hazard function $h(t)$ is the instantaneous rate that an event will happen within an infinitesimally small interval just after time $t$ given it has not happened yet up to time $t$, $ h(t) = \lim_{\Delta t \to\infty} \frac{Pr(t \leq T \leq t+\Delta t \ | \ T \geq t)}{\Delta t} = \frac{f(t)}{S(t)} $.

\noindent \textbf{Computing the likelihood of cascades.} The  activations for a cascade are represented by the sequence of tuples $(j, i, t_i)$ where $j, i$ $\in E_C$, $j$ is the parent of $i$ and $t_i$ is the timestamp at which $i$ was activated. Here the propagation happened from $j$ to $i$ denoting that $j$ = $parent(i)$ in the diffusion tree. The sequence of tuples are chronologically ordered by their activation times. All nodes $k$ $\in V_C$ such that $t_k \leq t_i$ and $k \neq j$, are termed as \textit{non-parent} nodes of $i$.  We note that $exp(v)$ is a subset of the set of non-parent nodes since the exposure nodes implicitly satisfy the constraints of non-parent nodes. 

Each directed edge in the cascade network $G_C$, $j$ $\rightarrow i$ is associated with a transmission function $f_{ji}(t_i|t_j)$, which is the conditional likelihood of the infection being spread to node $i$ at time $t_i$ given that the same infection has already affected node $j$ at time $t_j$. These functions are parametrized by parameters $\alpha_{ji}$ for every pair although later we will show how to remove the static nature of these pairwise unique transmission rates.  We focus on shift-invariant transmission functions whose value only depends on the time difference, i.e.\ , $f_{ji}(t_i|t_j)$ = $f_{ji}(t_i - t_j)$ = $f_{ji}(\Delta_{ji})$ where $\Delta_{ji} :=$ $t_i - t_j$. If there is no infection edge from $j$ to $i$, we will have $f_{ji}(\Delta_{ji})$ = $0$ and $h_{ji}(\Delta_{ji}) = 0$. For each directed diffusion link from $k$ $\rightarrow$ $i$, such that $k \in \zeta_{i, C}$, we associate an exposure function $E_{ki}$ parametrized by $\eta_k$ that quantifies the extent to which the exposure node $k$ affects the node $i$ in adopting the cascade $C$. An illustration of various node components in the diffusion mechanism is shown in Figure~\ref{fig:cas_inf}. 

\begin{figure}[!t]
	\includegraphics[width=8.5cm, height=2.6cm]{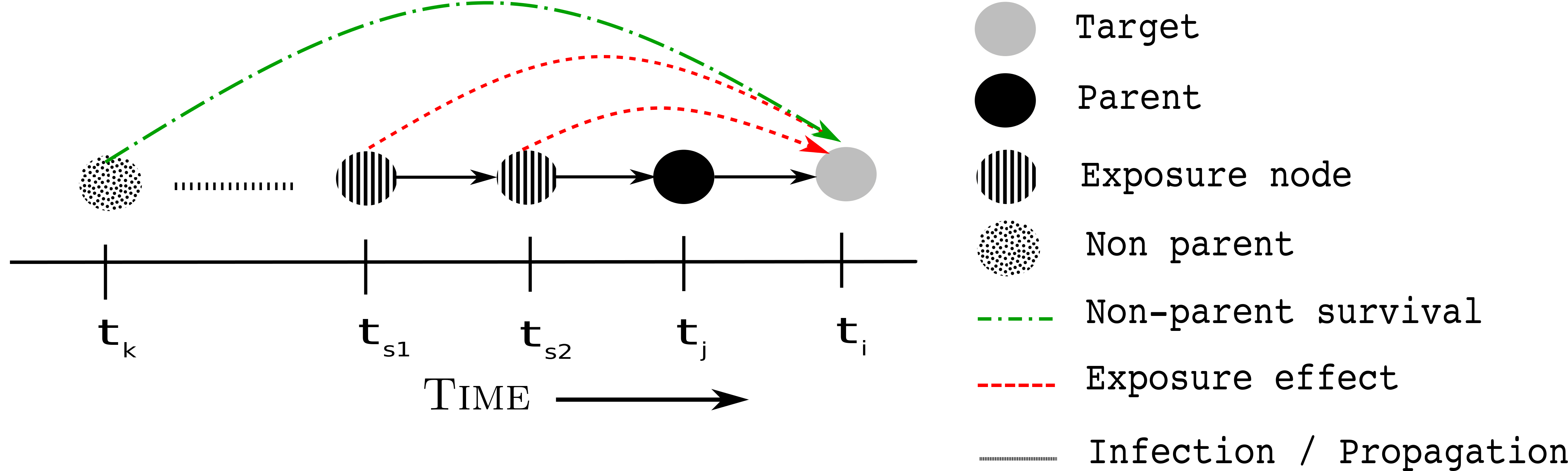}
	\caption{The components used for computing likelihood of node activations in a cascade.}
	\label{fig:cas_inf}
\end{figure}

The likelihood of the observed infections of a cascade $C$ can be formulated as $f(\mathbf{t}; \mathbf{A}, \mathbf{Z}) = \prod_{t_i \leq T_{\mathcal{R}}} f(t_i|t_1, \ldots, t_{\mathcal{R}}; \mathbf{A}, \mathbf{Z})$, where $A$ denotes the set of transmission parameters for each pair $u, i$, $\in E_C$, and $Z$ denotes the set of exposure parameters for each node $w \in \zeta_C$. Since we assume node infections are conditionally independent given the parents of the infected nodes, the likelihood factorizes over the nodes. The likelihood that the parent $j$ of an individual $i$ infected it at time $t$, $t \in [t_1, T_{\mathcal{R}}]$ and that $i$ survived from all other adopters of $C$ prior to time $t$ is given by 
\begin{equation*}
TL = f(t_i | t_j ; \alpha_{j, i}) \ \times \ \prod_{j \neq k; \ t_k < t_i} S(t_i | t_k; \alpha_{k, i})
\end{equation*}
To add the exposure effect on node $i$ from nodes in $exp(i)$, we have the exposure likelihood term 
$ EL \  =  \ \prod_{m \in \zeta_i} E(t_i | t_m; \eta_m) $. Then the likelihood of an individual cascade is defined as follows: 
 
\begin{multline*}
\mathbf{f(t^{\leq T_{\mathcal{R}}}; A, Z)}  = \prod_{t_i \leq T_{\mathcal{R}} } f(t_i | t_j ; \alpha_{j, i}) \ \times \ \prod_{j \neq k; \ t_k < t_i} S(t_i | t_k; \alpha_{k, i})  \times \\ \prod_{m \in \zeta_i} E(t_i | t_m; \eta_{m})
\end{multline*}

Note that the exposure term $E(.)$ is computed from only the exposure nodes extracted from the motifs on the subsequences - so the likelihood is computed subsequence wise in order of time. Finally, the likelihood of all the cascades is defined as $\prod_{c\in C} \mathbf{f(t^c; A, Z)}$.

\noindent \textbf{Centrality modulated Rayleigh model.} In this paper, we adopt the Rayleigh model for the transmission and the survival functions. Given a pair of nodes $j$ and $i$, the density, hazard and survival functions of Rayleigh distributions are:$f_{ji}(\Delta_{ji}) = \alpha_{ji}\ . \ \Delta_{ji}\ .\ exp\big(-\frac{1}{2 } \ . \ \alpha_{ji} \ . \ \Delta_{ji}^2\big)$, $h_{ji}(\Delta_{ji}) = \alpha_{ji}\ . \ \Delta_{ji}$ and $S_{ji}(\Delta_{ji}) =  exp\big(-\frac{1}{2 } \ . \ \alpha_{ji} \ . \ \Delta_{ji}^2\big)$ respectively. The drawback of directly learning the node specific and pairwise the above parameters $\mathbf{A}$, $\mathbf{Z}$ is that, firstly, for any pair $(j, i)$ of nodes in the test data for which the infection likelihood is calculated and not observed in training data, it would not be possible in any way to compute or infer them instantly. Secondly, following positive results in \cite{Cui2013}, we use the centrality of nodes from the historical diffusion/social network to decompose these parameters in terms of node attributes, so that the corresponding transmission or exposure parameters of any pair of nodes during test time can be inferred instantaneously, if the node attributes are known. We decompose the parameters as $\alpha_{ji} = x_i \ . \ x_j \ . \ \lambda$, where the symbol $x$ denotes the node attribute computed from the diffusion knowledge network and $\lambda $ denotes the corresponding scaling control parameter. The exposure function on a node $j$ is defined as $E_{j, m} =   \eta_m  \ . \ exp(- \ (t_j - t_m)) $ and the exposure parameter is defined accordingly as $\eta_m = x_m \ . \ \beta$.

\subsubsection{Objective Formulation.}
Considering $N$ distinct nodes, the likelihood optimization is formulated as follows:
\begin{equation}
\begin{aligned}
& \underset{A, Z}{\text{minimize}}
& &  -\sum_{c \in C} \mbox{log} \ f(t^c; A, Z) \\
& \text{subject to}
& & \alpha_{j, i} \geq 0, \forall i = 1 \ldots N, i \neq j \\
&&&  \eta_{m} \geq 0, \forall m = 1 \ldots N
\end{aligned}
\label{eq:opt_1}
\end{equation}
%
%\begin{theorem*}
%	Under the Rayleigh distribution for the survival and the hazard functions in the parameters of the pairwise transmission likelihoods, the joint optimization problem is convex in separate parameters \textbf{A} and \textbf{Z}.
%\end{theorem*}
We restrict the proof here  due to shortage of space but it can be decomposed into a sum of convex functions.
% Algorithm
\begin{algorithm}[t!]
	%\SetAlgoNoLine
	\KwIn{$G^{\tau'}$, $V^{\tau''}$, $\tau', \tau''$, $X$ the node attributes, $\lambda$, $\beta$}
	\KwOut{$E_{cut}$, the cut edges between $V^{\tau'}$ and $V^{\tau''}$ }
	
	$E_{cut}$ $\leftarrow$ $\phi$ \\
	\For{each node $v$ in $V^{\tau''}$}{
		$parent(v)$ $\leftarrow$ NULL, max\_LL $\leftarrow $ -INF\\
		\For{each pair $(s, e)$ $\in$  $E^{\tau'}$}{
			\tcc{Remove temporally far nodes}
			\If{ $t_{v, C} - t_{s, C}$  $>$ $t_{thresh}$} { 
				continue
			}
			$\alpha$ $\leftarrow$ $\lambda$ * $X[v]$ * $X[e]$, \ \ 
			$\eta$ $\leftarrow$ $\beta$ * $X[s]$
			
			Calculate Transmission LL $TL$ using $\alpha$, $s$, $v$ \\
			Calculate Exposure LL $EL$ using $\eta$, $e$

			\If{EL * TL $>$ max\_LL}{
				max\_LL $\leftarrow$ EL * TL\\
				$parent(v)$ $\leftarrow$ s
			}
			
		}
		$E_{cut}$ $\leftarrow$ $E_{cut}$ $\cup$ $(parent(v), v)$
	}
	
	return $E_{cut}$
	
	\caption{Inferring the cut edges: \textsc{InferCut}}
	\label{alg:edge_infer}
\end{algorithm}
Considering the centrality modulated Rayleigh model, the final objective function to be optimized can be written as
\begin{equation}
\label{eq:final_obj}
G(\textbf{A}, \textbf{Z}, \lambda, \beta) = G_1(\textbf{A}, \textbf{Z}) + n G_2(\lambda) + m G_3(\beta) 
\end{equation}
\begin{equation}
\label{eq:likelihood}
G_1(\textbf{A}, \textbf{Z}) = - \mbox{log} \ L(\textbf{A}, \textbf{Z})
\end{equation}
\begin{equation}
\label{eq:lambda}
G_2(\lambda) = \frac{1}{2P} \lVert \alpha - X \circ X \ . \ \lambda \rVert + \gamma_{\lambda} \lVert \lambda \rVert_1 
\end{equation}
\begin{equation}
\label{eq:beta}
G_3(\beta) = \frac{1}{2P} \lVert \eta - X \ . \ \beta \rVert + \gamma_{\beta} \lVert \beta \rVert_1
\end{equation}

where $P$ denotes the number of unique pairs of infections in the cascades. We prove that the final objective function in Equation~\ref{eq:final_obj} has a global minimum which we have included in the Appendix Section. Following this, we use standard Lasso solvers for Equation ~\ref{eq:lambda} and Equation~\ref{eq:beta} and the gradient descent with co-ordinate descent strategy method for solving Equation ~\ref{eq:likelihood}.

\subsection{Parent Inference}
Finally for the prediction problem, we consider a test subsequence $\tau''$ in the cascade lifecycle. For inferring the set $Parents(V^{\tau''})$ for a subsequence $\tau''$, we consider nodes in $G^{\tau'}$ in the preceding interval $\tau'$, whose timestamps, node attributes and edges are taken as input as shown in Figure~\ref{fig:pr_state}, and another set of nodes  $V^{\tau''}$ with known node attributes (or sampled from distribution) whose timestamps are known. Our main task is to infer $Parents(V^{\tau''})$ with $v' \in V^{\tau'}$, $\forall v' \in Parents(V^{\tau''})$  and return the cut-edges $E_{cut}$ between these two sets of nodes. The main difficulty of this case of dynamic network inference is that as an inference task, we are not predicting the diffusion edges that might have occurred within $V^{\tau''}$, an example of which is shown in Figure~\ref{fig:pr_state} marked by blue edges. Due to this condition, the number of actual cut edges is low since many nodes in $V^{\tau''}$ would be activated because of the influence of nodes within $V^{\tau''}$ itself or due to nodes activated prior to $\tau'$ and not due to nodes in $V^{\tau'}$. We present the algorithm for inference in Algorithm~\ref{alg:edge_infer}. For inferring $parent(v)$, $v \in V^{\tau''}$, we form 3-size motifs using $s, e, v$, the edge $(s, e)$ $\in E^{\tau'}$. That is, for each such directed edge, we consider $s$ as potential parent and $e$ as exposure node and calculate the product of likelihoods. We take that pair $(s, e)$ for $v$ which has the highest product and consider the corresponding $s$ as the parent.
\vspace{-.1in}
\section{Experimental Evaluation}
\label{sec:results}
In this section, we first describe the statistics of the dataset used for learning and evaluation. We perform a suite of experiments to evaluate our model with respect to some state-of-the art methods that are able to handle these kind of problems.
\subsubsection{Dataset.}
\label{sec:dataset}
To build the historical diffusion network, we use the dataset provided by WISE 2012 Challenge\footnote{http://www.wise2012.cs.ucy.ac.cy/challenge.html}. The dataset contains user data and the reposting information of each microblog along with the reposting times which enables us to form the cascades for each microblog separately. The diffusion network  $G_{\mathcal{H}}$ mentioned in the \textit{Preliminaries Section} is created by linking any two users who are involved in a microblog reposting action within the period May 1, 2011 and July 31, 2011. The primary purpose of $\mathcal{H}$ is two fold: acting as a social network to exact possible diffusion paths in future and for computing the node attributes $x$ used in the model.  This historical diffusion network has 6,470,135 nodes and 58,308,645 edges where the average node degree is 18.02. From the corpus of cascades which spanned between June 1, 2011 and August 31, 2011, we only work with cascades with more than 300 nodes in order to have sufficient information to evaluate the temporal networks in an incremental fashion. We collect a total of 4000 cascades in the process and we train using 3000 samples and put 1000 cascades aside for testing.\\
\textbf{Distribution of historical interactions.} As a first experiment, we measure the following metric $A_{v2u}$ as has been defined in \cite{goyal2010learning} as  the number of times $u$ reshared a message from $v$ irrespective of the message and it happened in history. For each user $u$ in the training set that participated in some cascade in the training set, we compute the distributions of $A_{v2u}$ for two sets of $v$: one where $v$ in $exp(u)$ for the cascade in consideration and the other $v$ in $nbrs(u) \ exp(u)$ where $nbrs(u)$ denote the neighbors of a user $u$ in $G_{\mathcal{H}}$. Figure~\ref{fig:dist} shows the density distributions of this measure for the two sets. It reveals that \textit{the nodes which motifs prune as exposure nodes are also the ones who are responsible for more propagation compared to the nodes which are just neighbors but had been pruned out by motifs.}\\
\textbf{Baseline Methods.} For evaluating the performance of the inference algorithm, let $E_{cut}$ be the set of actual edges between $V^{\tau'}$ and $V^{\tau''}$ and $\hat{E}_{cut}$ be the set of edges inferred by our algorithm. We define the following metrics of the classification model as follows: $ Precision = \frac{1}{|C|} \sum_{c \in C} \frac{|E_{cut, c} \cap \hat{E}_{cut, c}|}{|\hat{E}_{cut, c}|}$ and \\ $Recall = \frac{1}{|C|} \sum_{c \in C} \frac{|E_{cut, c} \cap \hat{E}_{cut, c}|}{|E_{cut, c}|}$. 
Simply put, we measure what fraction of the predicted edges are part of the ground truth edges as precision, and what fraction of the ground truth edges are part of the predicted edges as recall. Since for each node in $V^{\tau''}$, we try to predict  and evaluate its parent in $V^{\tau'}$ only when the ground truth parent is in $V^{\tau'}$ while ignoring the rest of the nodes in $V^{\tau''}$, we evaluate the performance of the models on recall. So, for every node in $V^{\tau''}$, every model gets only one chance to select one node among the nodes in $V^{\tau'}$ and therefore the number of predicted edges is equal to the number of ground truth edges making the precision and recall same for evaluation unlike a network inference problem. 
\begin{figure}[!t]
	\centering
	\includegraphics[width=7cm, height=4cm]{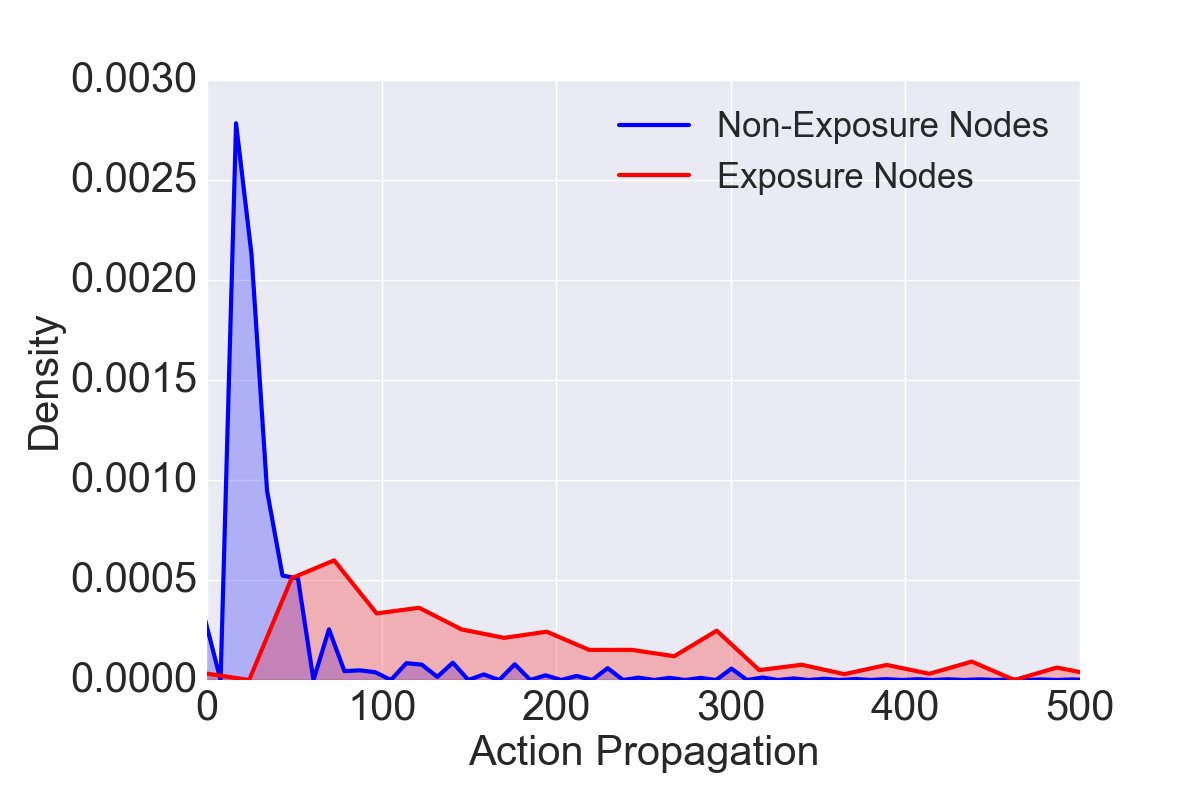}
	\caption{Distribution of $A_{v2u}$ - the propagation count from user $v$ to user $u$.}
	\label{fig:dist}
\end{figure}
Since there are no direct methods which have the same setting as our problem, we compare our framework with some recent advances in the field of information diffusion models that predict cascade adoption in the presence of active neighbors and exposures. We replicate one method that looks at all active neighbors without any pruning techniques \cite{goyal2010learning} to model the probability of action adoption using the measure $A_{v2u}$ and a second method that uses a probabilistic version of the General Threshold \cite{kempe2003maximizing} but taking exposures into account. Additionally we consider variations of our own model. The following frameworks have been used for comparison: (1) \textbf{InferCut - No Exposure(NE)}: For this procedure, we follow the same \textsc{InferCut} algorithm except that in Algorithm~\ref{alg:edge_infer}, we replace the likelihood value by just the transmission likelihood and avoid the other factor of exposure likelihood. We perform this to examine whether the additional effect of exposures have any influence on adoption, (2) \textbf{ Bernoulli Model}: We use this model to find the most probable nodes in $V^{\tau'}$ that infects the nodes in $V^{\tau''}$. Under this model, the probability that a user $u$ influences a node $v$ is given by the MLE estimator of success probability which comes from a Bernoulli trial. It is computed as $\frac{A_{v2u}}{A_v}$, $A_v$ denoting the number of previous cascades $v$ has participated in, (3)  \textbf{Complex Contagion (CC) Model}: This model \cite{fink2016investigating} is a variation of the General Threshold Model where the likelihood of adoption increases with more active neighbors. To approximate threshold-like behavior, they used a version of the logistic sigmoid function to calculate join probability of exposures, and (4) \textbf{Random - Temporal}: For each node in $V^{\tau''}$, we randomly select a node in $V^{\tau'}$ as its parent node that is $\Delta t$ time units within its own occurrence. For temporal heuristics for consideration, we rule out nodes in $V^{\tau'}$ which are temporally far away from it. 

For both the Bernoulli Model and the CC model, we compute the measures for all possible parents for each target node and only consider the node among the parents which has the highest value in terms of that measure. \\
\begin{figure}[!h]
	\centering
	\minipage{0.4\textwidth}
	\includegraphics[width=6cm, height=4cm]{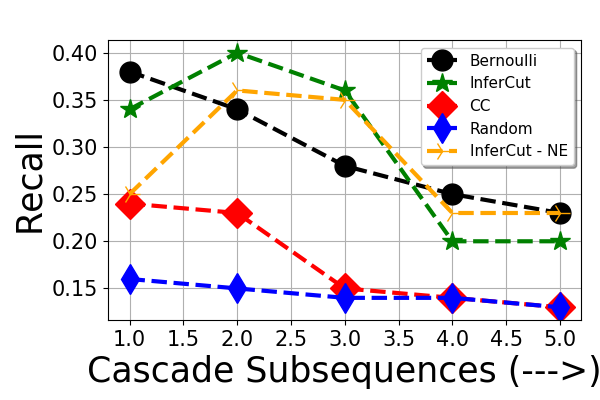}
	\subcaption{}
	\endminipage
	\hfill
	\\
	\minipage{0.4\textwidth}
	\includegraphics[width=6cm, height=4cm]{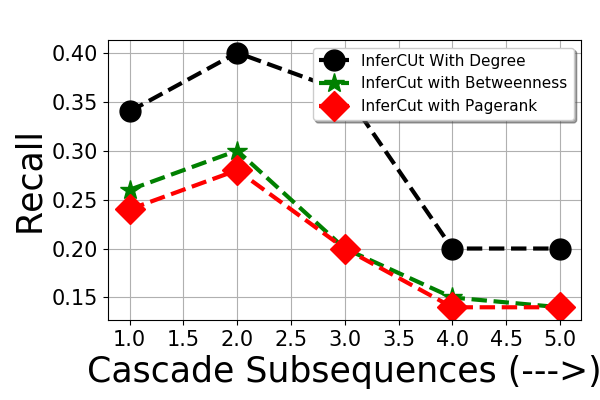}
	\subcaption{}
	\endminipage
	\hfill
	\caption{(a) Recall over the time subsequences from the start of the cascade, (b) Comparison of different centralities for the Rayleigh model.}
	\label{fig:acc}
\end{figure}

\noindent \textbf{Experimental setup and evaluation metrics.} For our cascade splitting method described in the \text{Preliminaries Section}, we fix $|V^{\tau'}|$ to 40 for all the cascades. That is to say, we consider a set of 40 nodes as parents when trying to predict the cut edges with the other endpoints being the susceptible target nodes in the succeeding subsequence. We use the \textsc{FANMOD} algorithm in \cite{wernicke2006fanmod} which is a widely used algorithm to find network motifs. For training we accumulate 3000 cascades and set aside an additional 1000 cascades for testing. These cascades have been selected by shuffling through the set randomly, yet we only perform the training one time due to the expensive nature of the computation. We set the parameters $m=5$ and $n=5$ in Equation~\ref{eq:final_obj} for all the evaluation results shown although we show a hyperparameter sensitivity towards the end. 
 We evaluate our models by considering progressing intervals/subsequences of nodes from the start for all the cascades leaving the first interval. For an interval $i \in [1, \mathcal{Q}_C]$ we consider the network induced in the preceding interval $i-1$ as the input to \textsc{InferCut} and for each node $v$ in $V^{\tau'_i}$, we predict $parent(v)$. As shown in Figure~\ref{fig:acc}, we consider 5 intervals from start for every cascade, and since we consider cascades over 300 reshares, these intervals are guaranteed to exist for all cascades in our dataset. \\
 \begin{figure}[!h]
 	\centering
 	\minipage{0.25\textwidth}
 	\includegraphics[width=4cm, height=3cm]{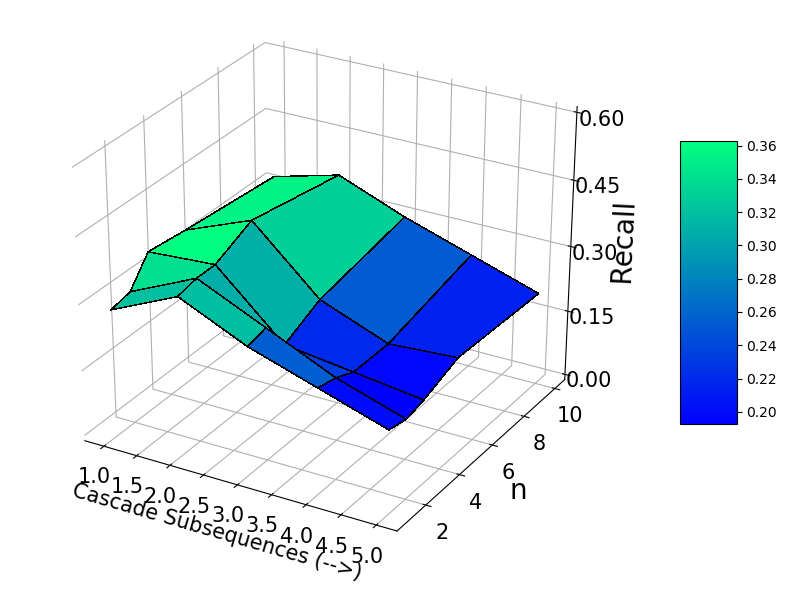}
 	\subcaption{}
 	\endminipage
 	\hfill
 	\minipage{0.22\textwidth}
 	\includegraphics[width=4cm, height=3cm]{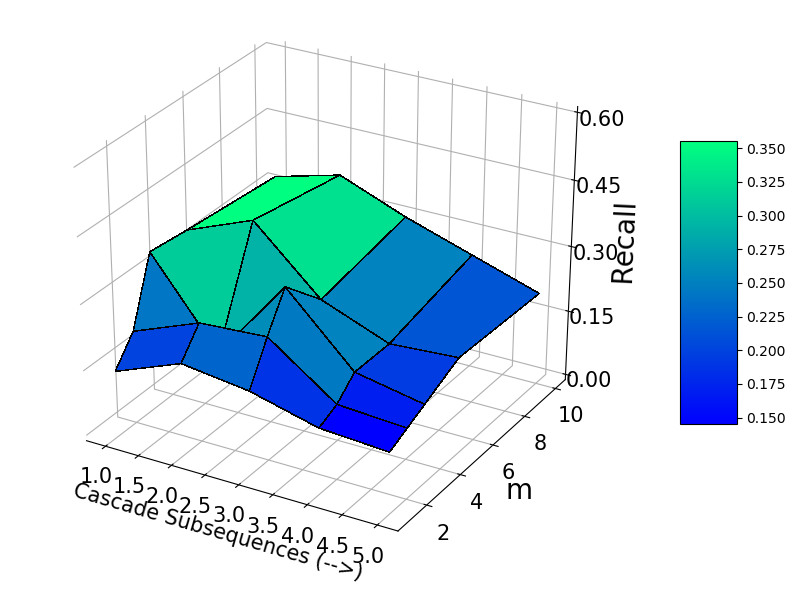}
 	\subcaption{}
 	\endminipage
 	\hfill
 	\caption{Sensitivity of the model w.r.t (a) $n$ that controls the effect of the transmission between node pairs, (b) $m$ that controls the effect of exposures.}
 	\label{fig:param_sens}
 \end{figure}

\noindent \textbf{Observations.} We present the results for all the models shown in Figure~\ref{fig:acc} from which we find that the low recall of the \textsc{Random - Temporal} approach suggests the challenging nature of this problem. The exposure effect with InferCut helps in obtaining a maximum recall in the initial stages of the cascades recovering almost 40\% of the actual edges in the second interval, although after the third interval, the performance degrades below the \textsc{No Exposure - InferCut} method. One reason for this performance is that over time, the topology of the cascade structure changes rapidly to an extent that it cannot be well represented by network motifs making the exposure effect play a very little role in the diffusion towards the end. We find that the sigmoid function in the CC model performs the worst among all the models apart from the random case, and one of the reasons being is that the model heavily relies on just the exposure count and not the time difference nor the node attributes - in a way just the exposure count information is not sufficient for diffusion processes. In addition, we experiment with different node attributes for the centrality modulated Rayleigh model where the transmission rate is parameterized by the node attributes of the edge nodes. We test with 3 centralities which are pre-computed based on $G_{\mathcal{H}}$: nodal degree, betweenness and Pagerank. From Figure~\ref{fig:acc}(b), we find that while the model performs the best with nodal degree, the Pagerank measure yields poor performance compared to the other two. This shows that a users' decision is largely impacted by the local network surrounding it and therefore the degree plays a larger role than measures which consider the global network structure.\\
\noindent \textbf{Hyperparameter Sensitivity.} We discuss the model sensitivity towards the hyperparameters $n$ and $m$ which control how much the transmission and  the exposure respectively impact the performance of prediction. From Figure~\ref{fig:param_sens}(b), we find that the performance in terms of recall gradually increases till the value of 5 after which it saturates, denoting that lowering the exposure hyperparameter that uses the node attribute values reduces the number of correctly recovered edges and thus establishing the importance of the node attributes in terms of the exposure hyperparameter. However in Figure~\ref{fig:param_sens}(a), we find that the transmission hyperparameter does not have much impact when varied over the ranges of $m$. 
\vspace{-.1in}
\section{Related Work}
Modeling the dynamics of network diffusion \cite{Gomez2011, sarkar2017understanding} has attracted widespread attention recently due to its applications in spread of epidemics.  Modeling the trajectory of the cascade growth has been studied in \cite{Cui2013}, where the authors build a model to understand the underlying diffusion process and in the event predict the future course of the cascade lifecycle.  The closest in the area of reconstructing epidemics over time has been done in \cite{Rozenshtein} where the authors use Steiner trees to infer the propagation structure. However in most real datasets, observing the cascade from the start is difficult due to time of observation \cite{sundareisan2015hidden} and missing data issues. In that context, our work only relies on the current information available to determine the source of infections. Similarly, a recent study \cite{stochastic_temporal} on recovering the stochastic temporal networks of diffusion rely on the null models of dyadic mutually independent interactions. Although external influence or cumulative exposures  \cite{external_inf} has been attributed to information diffusion in microblogging platforms, the absence of any concrete study of how exposures affect the diffusion rate has been one of motivations of this work.
\vspace{-.1in}
\section{Conclusion}
We evaluated the need and effectiveness of additional information available from the partial cascade structure, in determining source of infections for future nodes. We further showed how effective this information over time could be across different stages of the cascade lifecycle, while being able to prune relevant nodes for this additional information to keep the computation scalable. One of the directions in which this research can be extended is to be able to rank exposure nodes in terms of certain functions in order to be able to directly infer the source of infection with high certainty. \\
 
\noindent \textbf{Acknowledgements.} Some of the authors are supported through the  ARO grant W911NF-15-1-0282.

\bibliographystyle{abbrv}
\bibliography{main}

\appendix
\section{Optimizing the objective function}
Considering $N$ distinct nodes, the objective function to be optimized based on the likelihood is formulated as follows:
\begin{equation}
\begin{aligned}
& \underset{A, Z}{\text{minimize}}
& &  -\sum_{c \in C} \mbox{log} \ f(t^c; A, Z) \\
& \text{subject to}
& & \alpha_{j, i} \geq 0, \forall i = 1 \ldots N, i \neq j \\
&&&  \eta_{m} \geq 0, \forall m = 1 \ldots N
\end{aligned}
\label{eq:opt_1}
\end{equation}

The likelihood can be written as:

\begin{multline}
l(t^{\leq \mathcal{T_R}}; \ \mathbf{A}, \mathbf{Z}) = \prod_{t_i \leq \mathcal{T_R} } \big[f(t_i \ | \ t_j; \alpha_{j,i})  \ \times \\ 
\prod_{j \neq k; t_k <t_i} S(t_i \ | \ t_k; \alpha_{j,i})  \\
\times \prod_{m \in \zeta_{i}} E(t_i \ | \ t_m; \eta_m) ]
\end{multline}

Rearranging the terms, we can also write the likelihood as follows:

\begin{multline}
l(t^{\leq \mathcal{T_R}}; \ \mathbf{A}, \mathbf{Z}) = \prod_{t_i \leq \mathcal{T_R} } \big[h(t_i \ | \ t_j; \alpha_{j,i})  \ \times \\ 
\prod_{t_k <t_i} S(t_i \ | \ t_k; \alpha_{j,i})  \\
\times \prod_{m \in \zeta_{i}} E(t_i \ | \ t_m; \eta_m) ]
\end{multline}

\begin{theorem*}
	Under the Rayleigh distribution for the survival and the hazard functions in the parameters of the pairwise transmission likelihoods, the joint optimization problem is convex in separate parameters \textbf{A} and \textbf{Z}.
\end{theorem*}

\begin{proof}
	Keeping \textbf{Z} constant, the log likelihood function can be written as the sum of following components:
	
	\begin{multline}
	-\sum_{i:t_i \leq T_c} \sum_{k: t_k < t_i} \mbox{log} \ \  S(t_i | t_k; \alpha_{k, i}) = \\ -\sum_{i:t_i \leq T_c} \sum_{k: t_k < t_i} \alpha_{k, i} \frac{(t_i - t_k)^2}{2}
	\end{multline}
	and
	\begin{equation}
	-\sum_{i:t_i < T_c, j  \ = \  parent(i)} \mbox{log} \ \ h(t_i | t_j) = - \sum_{i:t_i < T_c} \mbox{log} \ \ \alpha_{j, i}(t_i - t_j)
	\end{equation}
\end{proof}

Since the sum of a linear and a convex function is convex, therefore the minimization problem in \textbf{A} is convex.

Similarly, we can show that the problem is convex in \textbf{Z} keeping constant \textbf{A}.\\

Considering the centrality modulated Rayleigh model, the final formulation can be written as

\begin{equation}
\label{eq:final_obj}
G(\textbf{A}, \textbf{Z}, \lambda, \beta) = G_1(\textbf{A}, \textbf{Z}) + n G_2(\lambda) + m G_3(\beta) 
\end{equation}

\begin{equation}
\label{eq:likelihood}
G_1(\textbf{A}, \textbf{Z}) = - \mbox{log} \ L(\textbf{A}, \textbf{Z})
\end{equation}
\begin{equation}
\label{eq:lambda}
G_2(\lambda) = \frac{1}{2P} \lVert \alpha - X \circ X \ . \ \lambda \rVert + \gamma_{\lambda} \lVert \lambda \rVert_1 
\end{equation}
\begin{equation}
\label{eq:beta}
G_3(\beta) = \frac{1}{2P} \lVert \eta - X \ . \ \beta \rVert + \gamma_{\beta} \lVert \beta \rVert_1
\end{equation}

\begin{theorem*}
	The final objective function in Equation~\ref{eq:final_obj} has a global minimum
\end{theorem*} 

\begin{proof}
	It is evident that the least squares formulation has a global minimum for Equations~\ref{eq:lambda} and Equation~\ref{eq:beta}. We need to prove that Equation~\ref{eq:likelihood} has a global minimum, that is the log likelihood has a global maximum. 
	
	Using the second order derivative computations, it can be shown that  $\frac{\partial^2 L }{\partial \alpha_{j, i}^2}$ $<0$ and  $\frac{\partial^2 L }{\partial \eta_{j}^2}$ $< 0$.
	Since $\frac{\partial^ L }{\partial \alpha_{j, i}^2}$ $<0$ and  $\frac{\partial^2 L }{\partial \eta_{j}^2}$ $< 0$, the conditional marginal posterior densities of $\alpha_{j, i}$ and $\eta_k$ are log-concave. Moreover, for specific ranges of $\alpha_{j, i}$ and $\eta_k$, the values of  $\frac{\partial L }{\partial \alpha_{j, i}}$ $>0$ and $\frac{\partial L }{\partial \eta_k}$ $>0$ and for certain other ranges, $\frac{\partial L }{\partial \alpha_{j, i}}$ $<0$ and $\frac{\partial L }{\partial \eta_k}$ $<0$. This means that there should be a global maximum of L and hence the final objective formulation should have a global minimum.
\end{proof}

Since global minimum is guaranteed, we use a coordinate descent algorithm that uses the following over multiple iterations:

\begin{equation*}
\alpha^{[it+1]} = argmin_\alpha G(\alpha, \eta^{[it]},\lambda^{[it]}, \beta^{[it]}) 
\end{equation*} 

\begin{equation*}
\eta^{[it+1]} = argmin_\eta G(\alpha^{[it+1]}, \eta,\lambda^{[it]}, \beta^{[it]}) 
\end{equation*}

\begin{equation*}
\lambda^{[it+1]} = argmin_\lambda G(\alpha^{[it+1]}, \eta^{[it+1]},\lambda, \beta^{[it]})
\end{equation*}

\begin{equation*}
\beta^{[it+1]} = argmin_\beta G(\alpha^{[it+1]}, \eta^{[it+1]},\lambda^{[it+1]}, \beta) 
\end{equation*}

\section{Extracting Exposure Nodes using Network Motifs}
In this section, we give a detailed algorithm for pruning the exposure nodes $exp(v)$, $v \in
V^{\tau''} $from the temporal network $G^{\tau'} \cup G^{\tau''}$, where $\tau'$ precedes $\tau''$ in the order of the subsequences in the cascade in consideration. Note that these exposure nodes for a node is compute cascade wise.
\begin{algorithm}[!h]
	%\SetAlgoNoLine
	\KwIn{Verties $V^{\tau''}$, Motif patterns $\mathcal{M}_{E}$, temporal network $G^N$ =  $G^{\tau'}_C$ $\cup$ $G^{\tau''}$ = $(V^N, E^N)$, threshold $th$ for the AND gate constraint, size $k$ of motifs. }
	\KwOut{$\zeta_{v, C}$ for all $v \in V^{\tau''}$}
	
	$\mathcal{M}_C^N, \ \{M_C^N\}$ = compute\_motifs \ ($G^N$, size=$k$) \tcp*{M denotes all patterns found in $G^N$, \{M\} denotes the list of 3-vertex maps for the corresponding patterns} 
	
	\For{each pattern $M$ $\in$  $\mathcal{M}_C^N$}{
		\If{M is not ismorphic to any pattern in $\mathcal{M}_{E}$ }{
			Skip this pattern and continue
		}
		\For{each $tgt$ in $V^{\tau''}$}{
			$src$ $\leftarrow$ $ parent(tgt)$ \;
			$\zeta_{tgt, C}$ $\leftarrow$ $\phi$\;
			$\delta_{tgt, C}$ $\leftarrow$ $\{\}$ \tcp*{empty hash table data structure}
			\For{each $m$ $\in $ $\{M_C^N\}$}{
				\If{src or tgt is not part of vertices in $m$}{
					Skip this instance and continue
				}
				
				$tv$ = $V(m)$ $ \setminus $ $\{src, tgt\}$ \tcp*{third vertex in $m$ barring the source and target pair} 
				\If{$t_{tv}$ $<$ $t_{tgt}$}{
					$\zeta_{tgt, C}$ $\leftarrow$ $\zeta_{tgt, C}$ $\cup$ $tv$ \;
					$\delta_{tgt, C}[tv]$ $\leftarrow$ $\frac{1}{t_{tgt} - t_{tv}}$\;
					
				} 
				
			}
		}
		
		\tcc{AND gating constraint: iterate over all the time elemenst until their product crosses the threshold }
		prod $\leftarrow$ multiply \ ($\delta_{tgt, C}.values()$)   \tcp*{$prod$ holds the product of all elements of $\delta_{tgt, C}$ }
		sorted\_map $\leftarrow$ sort $\delta_{tgt, C}$ by its elements \;
		\While{prod $\leq$ $th$ }{
			Delete the first element of sorted\_map \;
			prod $\leftarrow$ multiply \ ($sorted\_map.values()$)\;
		}
	}
	$\zeta_{tgt, C}$ $\leftarrow$ $sorted\_map.keys()$\;
	
	return $\zeta_C$
	
	\caption{Computing $\zeta_{v, C}$}
	\label{alg:zeta_comp}
\end{algorithm}

\end{document}